\documentclass[a4paper]{article}

\usepackage[latin1]{inputenc} %
\usepackage[T1]{fontenc} %
\usepackage{RR}
\usepackage{hyperref}

\usepackage{graphicx,amsmath}
\usepackage{subfigure}
\usepackage{epsfig}
\usepackage{verbatim}
\usepackage{algorithm}
\usepackage{algorithmic}
\usepackage{color}
\usepackage{multirow}
\usepackage{latexsym}

\newcommand{\beq}{\begin{equation}}
\newcommand{\deq}{\end{equation}}

\newcommand{\beqm}{\begin{equation*}}
\newcommand{\deqm}{\end{equation*}}

\newcommand{\baq}{\begin{eqnarray}}
\newcommand{\daq}{\end{eqnarray}}

\newcommand{\baqm}{\begin{eqnarray*}}
\newcommand{\daqm}{\end{eqnarray*}}

\newcommand{\super}[1]{^{\mathrm{#1}}}

\newcommand{\qed}{\rightline{$\Box$}}

\newcommand{\one}{{\bf 1}}
\newcommand{\range}[2]{\in#1,\dots,#2}

\newenvironment{proof}{{\bf Proof. }}{\hfill \rightline{$\Box$}\medskip}
\newtheorem{theorem}{Theorem}[section]

\newtheorem{definition}{Definition}[section]

\newtheorem{algo}{Algorithm}[section]

\def\t1{\hbox{\bf 1}}
\def\P{\hbox{\sf P}}

\RRdate{August 2010}

\RRauthor{
Konstantin Avrachenkov\thanks{INRIA Sophia Antipolis-M\'editerran\'ee, France, K.Avrachenkov@sophia.inria.fr}
\and
Nelly Litvak\thanks{University of Twente, The Netherlands, N.Litvak@ewi.utwente.nl}
\and
\\ Danil Nemirovsky\thanks{INRIA Sophia Antipolis-M\'editerran\'ee, France, Danil.Nemirovsky@gmail.com}
\and
Elena Smirnova\thanks{INRIA Sophia Antipolis-M\'editerran\'ee, France, Elena.Smirnova@sophia.inria.fr}
\and
Marina Sokol\thanks{INRIA Sophia Antipolis-M\'editerran\'ee, France, Marina.Sokol@sophia.inria.fr}
}

\authorhead{K. Avrachenkov, N. Litvak, D. Nemirovsky, E. Smirnova \& M. Sokol}

\RRtitle{Les M\'ethodes Monte Carlo pour Top-k Listes de PageRank
Personnalis\'e avec l'application a disambiguation de noms} 
\RRetitle{Monte Carlo Methods for Top-k Personalized PageRank Lists
and Name Disambiguation}
\titlehead{Monte Carlo Methods for Top-k Personalized PageRank Lists}

\RRresume{ Nous \'etudions le probl\`eme de d\'etection rapide de
top-k listes de PageRank Personnalis\'e. Ce probl\`eme a plusieurs
applications importantes telles que la recherche des coupes locales
de graphes, l'\'estimation de la distance de la similarit\'e, et
disambiguation de noms. En particulier, nous appliquons nos
resultats a construction des algorithmes efficaces pour le
probl\`eme de disambiguation de noms de personnes. Notre \'etude est
bas\'e sur les deux observations suivantes. D'abord, il est cruciale
que nous trouvons rapidement les top-k voisins les plus importants
d'un noeud. Cependant, l'ordre exact dans le top-K ainsi que les
valeurs exactes de PageRank sont de loin pas si cruciale.
Deuxiemement, un petit nombre de elements erron\'es dans les top-k
listes ne degrade pas vraiment la qualite des listes de top-k, mais
ce sacrifice am\'eliore significativement la performance des
algorithmes. Sur la base de ces deux observations cl\'es nous
proposons des m\'ethodes de type Monte Carlo pour la d\'etection
rapide de top-k listes de PageRank Personnalis\'e. Nous offrons
l'\'evaluation des performances des m\'ethodes propos\'ees et nous
donnons crit\`eres d'arr\^et. En suite, nous appliquons les
m\'ethodes au probl\`eme de disambiguation de noms de personnes.
Notre approche bas\'e sur PageRank Personnalis\'e et les m\'ethodes
Monte Carlo a recu le deuxi\`eme prix de la comp\'etion WePS 2010.}

\RRabstract{
We study a problem of quick detection of top-k Personalized PageRank
lists. This problem has a number of important applications such as
finding local cuts in large graphs, estimation of similarity distance
and name disambiguation. In particular, we apply our results to
construct efficient algorithms for the person name disambiguation
problem. We argue that when finding top-k Personalized PageRank lists two
observations are important. Firstly, it is crucial that we detect fast
the top-k most important neighbours of a node, while the exact order
in the top-k list as well as the exact values of PageRank are by far
not so crucial. Secondly, a little number of wrong elements in top-k
lists do not really degrade the quality of top-k lists, but it can lead to
significant computational saving. Based on these two key observations
we propose Monte Carlo methods for fast detection of top-k Personalized
PageRank lists. We provide performance evaluation of the proposed methods
and supply stopping criteria. Then, we apply the methods to the person name
disambiguation problem. The developed algorithm for the person name
disambiguation problem has achieved the second place in the WePS 2010
competition.
}

\RRmotcle{PageRank Personnalis\'e, M\'ethodes Monte Carlo,
Disambiguation de Noms de Personnes} \RRkeyword{Personalized
PageRank, Monte Carlo Methods, Person Name Disambiguation} 

\RRprojet{Maestro, Axis}  
\RRtheme{\THCom} 
\URSophia 

\begin{document}

\RRNo{7367}

\makeRR

\section{Introduction}
\label{sec:intro}

Personalized PageRank or Topic-Sensitive PageRank \cite{haveliwala02topic} is a generalization
of PageRank \cite{BPMW98}. Personalized PageRank is a stationary distribution of a random
walk on an entity graph. With some probability the
random walk follows an outgoing link with uniform distribution and with the complementary
probability the random walk jumps to a random node according to a personalization distribution.
Personalized PageRank has a number of applications. Let us name just a few. In the original
paper \cite{haveliwala02topic} Personalized PageRank was used to introduce the personalization
in the Web search. In \cite{citeulike:1286302,LK03,conf/icalt/Wissner-Gross06} Personalized
PageRank was suggested for finding related entities. In \cite{citeulike:2348659}
Green measure, which is closely related to Personalized PageRank, was suggested
for finding related pages in Wikipedia. In
\cite{conf/focs/AndersenCL06, conf/waw/AndersenCL07} Personalized PageRank was
used for finding local cuts in graphs and in \cite{ADNPS08} the Personalized
PageRank was applied for clustering large hyper-text document collections.
In many applications we are interested in detecting top-k elements with the
largest values of Personalized PageRank. Our present work on detecting top-k
elements is driven by the following two key observations:

\medskip

\noindent {\bf Observation 1:}
Often it is crucial that we detect fast the top-k elements with the largest values
of the Personalized PageRank, while the exact order in the top-k list as well as the exact
values of the Personalized PageRank are by far not so important.

\medskip

\noindent {\bf Observation 2:}
It is not crucial that the top-k list is determined exactly, and
therefore we may apply a relaxation that allows a small number of elements to be
placed erroneously in the top-k list. If the Personalized PageRank values of these
elements are of a similar order of magnitude as in the top-k list, then such
relaxation does not affect applications, but it enables us to take advantage of the
generic ``80/20 rule'': 80\% of the result is achieved with 20\% of efforts.

\medskip

We argue that the Monte Carlo approach naturally takes into
account the two key observations. In \cite{breyer02markovian} the Monte Carlo approach
was proposed for the computation of the standard PageRank. The estimation of the
convergence rate in \cite{breyer02markovian} was very pessimistic. Then, the implementation
of the Monte Carlo approach was improved in \cite{FRCS05} and also applied there to
Personalized PageRank. Both \cite{breyer02markovian} and \cite{FRCS05} only use end points
as information extracted from the random walk. Moreover, the approach of \cite{FRCS05}
requires extensive precomputation efforts and is very demanding in storage resource.
In \cite{ALNO07}, it has been shown that to find elements with large values of PageRank
the Monte Carlo approach requires about the same number of operations as one iteration
of the power iteration method. In the present work we show that to detect top-k list
of elements when $k$ is not large we need even smaller number of operations. In our test
on the Wikipedia entity graph with about 2 million nodes we have observed that typically
few thousands of operations are enough to detect the top-10 list with just two or three
erroneous elements. Namely to detect a relaxation of the top-10 list we spend just
about 1-5\% of operations required by one power iteration. In the present work we
provide theoretical justification for such a small amount of required operations.
We also apply the Monte Carlo methods for Personalized PageRank to the person name
disambiguation problem. Name resolution problem consists in clustering search results for a person name according to found namesakes. We found that considering patterns of Web structure for name resolution problem results in methods with very competitive performance.



%

\section{Monte Carlo methods}
\label{sec:methods}

Given a directed or undirected graph connecting some entities, the Personalized
PageRank $\pi(s,c)$ with a seed node $s$ and a damping parameter $c$ is defined as a
solution of the following equations
$$
\pi(s,c)=c \pi(s,c) P + (1-c) \one_s^T,
$$
$$
\sum_{j=1}^n \pi_j(s,c) = 1.
$$
where $\one_s^T$ is a row unit vector with one in the $s\super{th}$ entry and all
the other elements equal to zero, $P$ is the transition matrix associated with
the entity graph and $n$ is the number of entities. Equivalently, the Personalized PageRank
can be given by the explicit formula \cite{LM06,moler}
\beq
\label{chap:top-k:ppr_expl} \pi(s,c) = (1-c) \one_s^T [I-cP]^{-1}.
\deq
Whenever the values of $s$ and $c$ are clear from the context we
shall simply write $\pi$.

We would like to note that often the Personalized PageRank is defined
with a general distribution $v$ in place of $\one_s^T$. However, typically
distribution $v$ has a small support. Then, due to linearity, the problem
of Personalized PageRank with distribution $v$ reduces to the problem
of Personalized PageRank with distribution $\one_s^T$ \cite{JW03}.

In this work we consider two Monte Carlo algorithms. The first algorithm is
inspired by the following observation. Consider a random walk $\{X_t\}_{t
\ge 0}$ that starts from node $s$, i.e, $X_0=s$. Let at each step the
random walk terminate with probability $1-c$ and make a transition according to the
matrix $P$ with probability $c$. Then, the end-points of such a random walk
has the distribution $\pi(s,c)$.

\begin{algo}[MC End Point]
\label{chap:top-k:algo:endpoint}
Simulate $m$ runs of the random walk $\{X_t\}_{t \ge 0}$ initiated at node
$s$. Evaluate $\pi_j$ as a fraction of $m$ random walks which end at node
$j\range{1}{n}$.
\end{algo}

The next observation leads to another Monte Carlo algorithm for Personalized
PageRank. Denote $Z:=[I-cP]^{-1}$. We have the following interpretation for the
elements of matrix $Z$: $z_{sj}=E_s[N_j]$, where $N_j$ is the number of visits to
node $j$ by a random walk before a restart, and $E_s[\cdot]$ is the expectation
assuming that the random walk started at node $s$. Namely, $z_{sj}$ is the expected
number of visits to node $j$ by the random walk initiated at state $s$ with the run time
geometrically distributed with parameter $c$. Thus, the formula
\eqref{chap:top-k:ppr_expl} suggests the following estimator for Personalized
PageRank
\beq
\label{chap:top-k:ppr_hat}
\hat{\pi}_j(s,c) = (1-c) \frac{1}{m} \sum_{r=1}^m N_j(s,r),
\deq
where $N_j(s,r)$ is the number of visits to state $j$ during the run $r$ of the
random walk initiated at node $s$. Thus, we can suggest the second Monte Carlo
algorithm.

\begin{algo}[MC Complete Path]
\label{chap:top-k:algo:completepath}
Simulate $m$ \ runs of the random walk $\{X_t\}_{t \ge 0}$ initiated
at node $s$. Evaluate $\pi_j$ as the total number of visits to node $j$
multiplied by $(1-c)/m$.
\end{algo}

As outputs of the proposed algorithms we would like to obtain with high
probability either a {\em top-k list} of nodes or a {\em top-k basket} of nodes.

\begin{definition}
The \emph{top-$k$ list} of nodes is a list of $k$ nodes with largest
Personalized PageRank values arranged in a descending order of their Personalized PageRank values.
\end{definition}

\begin{definition}
The \emph{top-$k$ basket} of nodes is a set of $k$ nodes with largest
Personalized PageRank values with no ordering required.
\end{definition}

%

It turns out that it is beneficial to relax our goal and to obtain a top-$k$ basket
with a small number of erroneous elements.

\begin{definition}
We call \emph{relaxation-$l$} top-$k$ basket a realization when we allow at
most $l$ erroneous elements from top-$k$ basket.
\end{definition}

In the present work we aim to estimate the numbers of random walk runs $m$ sufficient
for obtaining top-$k$ list or top-$k$ basket or relaxation-$l$ top-$k$ basket with
high probability. In particular, we demonstrate that ranking converges considerably
faster than the values of Personalized PageRank and that a relaxation-$l$ with quite
small $l$ helps significantly.

Let us begin the analysis of the algorithms with the help of an illustrating example
on the Wikipedia entity graph. We shall carry out the development of the example throughout
the paper. There is a number of reasons why we have chosen the Wikipedia entity graph.
Firstly, the Wikipedia entity graph is a non-trivial example of a complex network.
Secondly, it has been shown that the Green's measure which is closely related to Personalized
PageRank is a good measure of similarity on the Wikipedia entity graph \cite{citeulike:2348659}.
In addition, we note that Personalized PageRank is a good similarity measure also in
social networks \cite{LK03} and on the Web \cite{SMP08}.
Thirdly, we apply our person name disambiguation algorithm on the real Web for which we cannot
compute the real values of the Personalized PageRank. The Personalized PageRank can be
computed with high precision for the Wikipedia entity graph with the help of BVGraph/WebGraph
framework \cite{BV04}.

\bigskip

\noindent
{\bf Illustrating example:}
Since our work is concerned with application of Personalized PageRank to the name disambiguation
problem, let us choose a common name. One of the most common English names is Jackson.
We have selected three Jacksons who have entries in Wikipedia: Jim Jackson (ice hockey),
Jim Jackson (sportscaster) and Michael Jackson. Two Jacksons have even a common given name
and both worked in ice hockey, one as an ice hockey player and another
as an ice hockey sportscaster. In Tables~\ref{tab:PPRexactJJplayer}-\ref{tab:PPRexactMJsinger}
we provide the exact lists of top-10 Wikipedia articles arranged according to Personalized PageRank
vectors. In Table~\ref{tab:PPRexactJJplayer} the seed node for the Personalized PageRank is the
article {\tt Jim Jackson (ice hockey)}, in Table~\ref{tab:PPRexactJJcaster} the seed node is
the article {\tt Jim Jackson (sportscaster)}, and in Table~\ref{tab:PPRexactMJsinger} the seed
node is the article {\tt Michael Jackson}. We observe that each top-10 list identifies quite
well its seed node. This gives us hope that Personalized PageRank can be useful in the name
disambiguation problem. (We shall discuss more the name disambiguation problem in Section~\ref{sec:disamb}.)
Next we run the Monte Carlo End Point method starting from each seed node. We note that top-10
lists obtained by Monte Carlo methods also identify well the original seed nodes. It is interesting
to note that to obtain a relaxed top-10 list with two or three erroneous elements we need different
number of runs for different seed nodes. To obtain a good relaxed top-10 list for {\tt Michael Jackson}
we need to perform about 50000 runs, whereas for a good relaxed top-10 list for {\tt Jim Jackson (ice hockey)}
we need to make just 500 runs. Intuitively, the more immediate neighbours a node has, the larger
number of Monte Carlo steps is required. Starting from a node with many immediate neighbours the Monte
Carlo method easily drifts away. In Figures~\ref{fig:jjplayermcep}-\ref{fig:mjsingermcep} we present
examples of typical runs of the Monte Carlo End Point method for the three different seed nodes.
An example of the Monte Carlo Complete Path method for the seed node {\tt Michael Jackson} is given
in Figure~\ref{fig:mjsingermccp}. Indeed, as expected, it outperforms the Monte Carlo End Point method.
In the following sections we shall quantify all the above qualitative observations.

{\small

\begin{table}[ht]
\caption{Top-10 lists for Jim Jackson (ice hockey)}
\centering
\begin{tabular}{|l|l|l|}
  \hline
  No. & Exact Top-10 List & MC End Point (m=500)\\ [0.5ex]
  \hline
  1 & Jim Jackson (ice hockey) & Jim Jackson (ice hockey) \\
  2 & Ice hockey & Winger (ice hockey) \\
  3 & National Hockey League & 1960 \\
  4 & Buffalo Sabres & National Hockey League \\
  5 & Winger (ice hockey) & Ice hockey \\
  6 & Calgary Flames & February 1 \\
  7 & Oshawa & Buffalo Sabres \\
  8 & February 1 & Oshawa \\
  9 & 1960 & Calgary Flames \\
  10 & Ice hockey rink & Columbus Blue Jackets \\
  \hline
\end{tabular}
\label{tab:PPRexactJJplayer}
\end{table}

}

{\small

\begin{table}[ht]
\caption{Top-10 lists for Jim Jackson (sportscaster)}
\centering
\begin{tabular}{|l|l|l|}
  \hline
  No. & Exact Top-10 List & MC End Point (m=5000)\\ [0.5ex]
  \hline
  1 & Jim Jackson (sportscaster) & Jim Jackson (sportscaster) \\
  2 & Philadelphia Flyers & Steve Coates \\
  3 & United States & New York \\
  4 & Philadelphia Phillies & United states \\
  5 & Sportscaster & Philadelphia Flyers \\
  6 & Eastern League (baseball) & Gene Hart \\
  7 & New Jersey Devils & Sportscaster \\
  8 & New York - Penn League & New Jersey Devils \\
  9 & Play-by-play & Mike Emrick \\
  10 & New York & New York - Penn League \\
  \hline
\end{tabular}
\label{tab:PPRexactJJcaster}
\end{table}

}

{\small

\begin{table}[ht]
\caption{Top-10 lists for Michael Jackson}
\centering
\begin{tabular}{|l|l|l|}
  \hline
  No. & Exact Top-10 List & MC End Point (m=50000)\\ [0.5ex]
  \hline
  1 & Michael Jackson & Michel Jackson\\
  2 & United States & United states\\
  3 & Billboard Hot 100 & Pop music\\
  4 & The Jackson 5 & Epic Records\\
  5 & Pop music & Billboard Hot 100\\
  6 & Epic records & Motown Records\\
  7 & Motown Records & The Jackson 5 \\
  8 & Soul music & Singing\\
  9 & Billboard (magazine) & Hip Hop music\\
  10 & Singing & Gary, Indiana\\
  \hline
\end{tabular}
\label{tab:PPRexactMJsinger}
\end{table}

}

\begin{figure}
    \centering
    \includegraphics[scale=0.5]{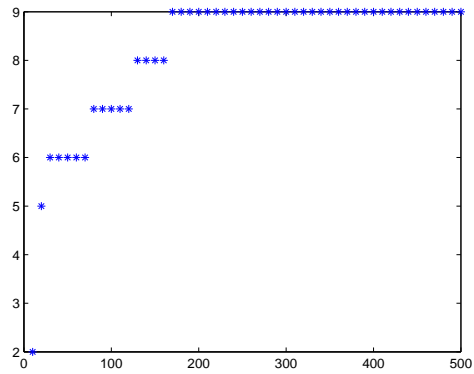}
    \caption{The number of correctly detected elements
    by MC End Point for the seed node {\tt Jim Jackson (ice hockey)}.}
    \label{fig:jjplayermcep}
\end{figure}

\begin{figure}
    \centering
    \includegraphics[scale=0.5]{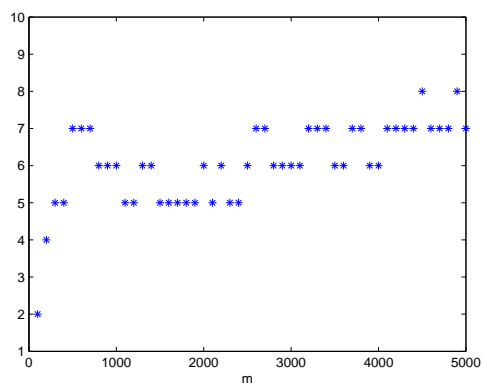}
    \caption{The number of correctly detected elements
    by MC End Point for the seed node {\tt Jim Jackson (sportscaster)}.}
    \label{fig:jjcastermcep}
\end{figure}

\begin{figure}
    \centering
    \includegraphics[scale=0.5]{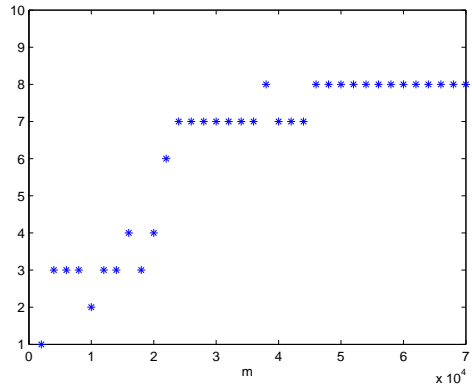}
    \caption{The number of correctly detected elements
    by MC End Point for the seed node {\tt Michael Jackson}.}
    \label{fig:mjsingermcep}
\end{figure}

\begin{figure}
    \centering
    \includegraphics[scale=0.5]{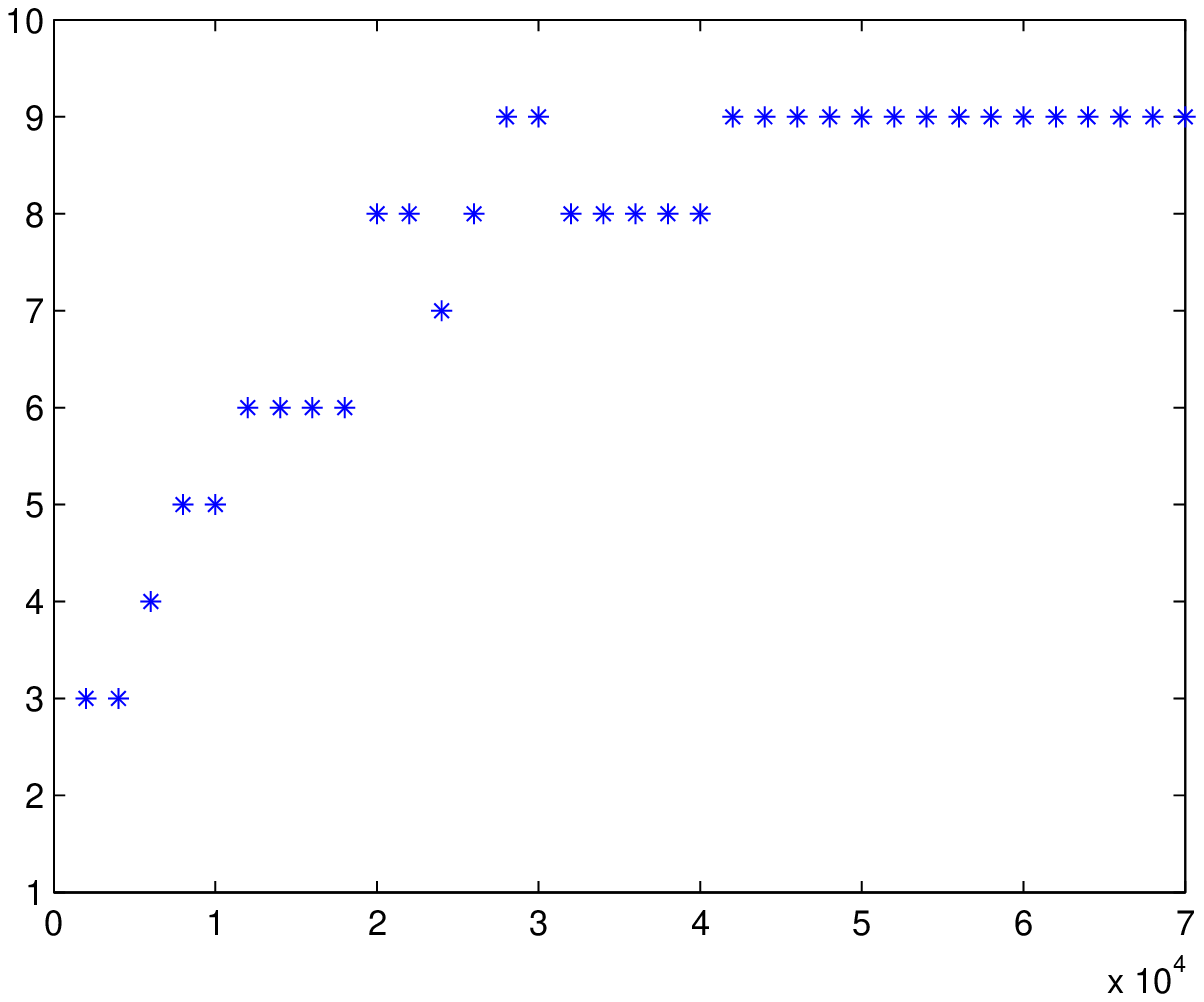}
    \caption{The number of correctly detected elements
    by MC Complete Path for the seed node {\tt Michael Jackson}.}
    \label{fig:mjsingermccp}
\end{figure}

\section{Variance based performance comparison and CLT approximations}
\label{sec:var}

In the MC End Point algorithm the distribution of end points is multinomial \cite{GVK213193329}.
Namely, if we denote by $L_j$ the number of paths that end at node $j$ after $m$
runs, then we have
\beq
\label{chap:top-k:multinom}
P\{L_1=l_1,L_2=l_2,\dots,L_n=l_n\}=\frac{m!}{l_1!l_2!\cdots l_n!}
\pi_1^{l_1} \pi_2^{l_2} \cdots \pi_n^{l_n}.
\deq
Thus, the standard deviation of the MC End Point estimator for
the $k\super{th}$ element is given by
\beq\label{chap:top-k:MCEPstd}
\sigma(\hat \pi_k)=\sigma(L_k/m)=\frac{1}{\sqrt{m}}\sqrt{\pi_k(1-\pi_k)}.
\deq

An expression for the standard deviation of the MC Complete Path is
more complicated. From (\ref{chap:top-k:ppr_hat}), it follows that
\beq
\label{chap:top-k:MCCPstd0}
\sigma(\hat \pi_k)=\frac{(1-c)}{\sqrt{m}}\sigma(N_k)
=\frac{(1-c)}{\sqrt{m}}\sqrt{E_s\{N_k^2\}-E_i\{N_k\}^2}.
\deq
First, we recall that
\beq
\label{chap:top-k:ENk}
E_s\{N_k\}=z_{sk}=\pi_k(s)/(1-c).
\deq
Then, from \cite{GVK232428700}, it is known that the second moment of $N_k$ is
given by $$
E_s\{N_k^2\}=[Z(2Z_{dg}-I)]_{sk},
$$
where $Z_{dg}$ is a diagonal matrix having as its diagonal the diagonal of matrix $Z$
and $[A]_{ik}$ denotes the $(i,k)\super{th}$ element of matrix $A$. Thus, we
can write
$$
E_s\{N_k^2\}=\one_s^T Z (2Z_{dg}-I)\one_k=\frac{1}{1-c}\pi(s)(2Z_{dg}-I)\one_k
$$
\beq\label{chap:top-k:ENk2}
=\frac{1}{1-c}\left(\frac{1}{1-c}\pi_k(s)\pi_k(k)-\pi_k(s)\right).
\deq
Substituting~\eqref{chap:top-k:ENk} and~\eqref{chap:top-k:ENk2}
into~\eqref{chap:top-k:MCCPstd0}, we obtain
\beq
\label{chap:top-k:MCCPstd1}
\sigma(\hat \pi_k)=\frac{1}{\sqrt{m}} \sqrt{\pi_k(s)(2\pi_k(k)-(1-c)-\pi_k(s))}.
\deq
Since $\pi_k(k)\approx 1-c$, we can approximate $\sigma(\hat \pi_k)$ with
$$
\sigma(\hat \pi_k) \approx \frac{1}{\sqrt{m}} \sqrt{\pi_k(s)((1-c)-\pi_k(s))}.
$$
Comparing the latter expression with \eqref{chap:top-k:MCEPstd}, we can see that MC End
Point requires approximately $1/(1-c)$ steps more than MC Complete Path. This was expected
as MC End Point uses only information from end points of the random walks. We would
like to emphasize that $1/(1-c)$ can be a significant coefficient. For instance, if
$c=0.85$, then $1/(1-c) \approx 6.7$.


Let us provide central limit type theorems for our estimators.

\begin{theorem}\label{chap:top-k:thm:CLTmultinom}
For large $m$, a multivariate normal density approximation to the multinomial
distribution $(\ref{chap:top-k:multinom})$ is given by
$$
f(l_1,l_2,\dots,l_n) =
\left(\frac{1}{2\pi m}\right)^{(n-1)/2} \times
$$
\beq\label{chap:top-k:multinom_approx}
\left(\frac{1}{n \pi_1 \pi_2 \cdots \pi_n}\right)^{1/2}
\exp\left\{ -\frac{1}{2}\sum_{i=1}^n \frac{(l_i-m \pi_i)^2}{m \pi_i}\right\},
\deq
subject to $\sum_{i=1}^n l_i = m$.
\end{theorem}
\noindent
\begin{proof}
See~\cite{conf/multivariateanalysis/KhatriM1969}
and~\cite{journal/jstor/TanabeM1992}.
\end{proof}

Now we consider MC Complete Path. First, we note that the vectors
$N(s,r)=(N_1(s,r),\dots,N_n(s,r))$ with $r=1,2,\dots$ form a sequence of i.i.d.
random vectors. Hence, we can apply the multivariate central limit theorem.
Denote
\beq
\hat N(s,m)=\frac{1}{m}\sum_{r=1}^m N(s,r).
\deq

\begin{theorem}\label{chap:top-k:thm:CLTiid}
Let $m$ go to infinity. Then, we have the following convergence in distribution
to a multivariate normal distribution
$$
\sqrt{m} \left(\hat N(s,m) - \bar{N} \right) \stackrel{D}{\longrightarrow} {\cal
N}(0,\Sigma(s)), $$
where $\bar{N}(s)=\one^T_s Z$ and
$\Sigma(s)=E\{N^T(s,r)N(s,r)\}-\bar{N}^T(s)\bar{N}(s)$ is a
covariance matrix, which can be expressed as
\beq
\label{chap:top-k:eq:covar_matrix_expr}
\Sigma(s)=\Omega\left(s\right)Z+Z^T\Omega\left(s\right)-\Omega\left(s\right)-Z^T\one_s\one^T_sZ.
\deq
where the matrix $\Omega(s)=\left\{\omega_{jk}(s)\right\}$ is defined by
\baqm
    \omega_{jk}(s) =
    \left\{
        \begin{array}{ll}
            z_{sj}, & \mbox{if $j=k$},\\
            0, & \mbox{otherwise}.
        \end{array}
    \right.
\daqm
\end{theorem}
\begin{proof}
The convergence follows from the standard multivariate central limit theorem.
We only need to establish the formula for the covariance matrix.

The covariance matrix can be expressed as follows~\cite{nemirovsky2009tensor}:
\beq
\label{chap:top-k:eq:covar_matrix_expr1}
\Sigma(s)=\sum_{j=1}^n
z_{sj}\left(D(j)Z+ZD(j)-D(j)\right)-Z^T\one_{s}\one^T_{s}Z,
\deq
where $D(j)$ is defined by
\beqm
d_{kl}(j)=
    \left\{
        \begin{array}{ll}
            1, & \mbox{if $k=l=j$},\\
            0, & \mbox{otherwise}.
        \end{array}
    \right.
\deqm

Let us consider $\sum_{j=1}^n z_{sj}D(j)Z$ in component form.
\baqm
\sum_{j=1}^n z_{sj}\sum_{\varphi=1}^{n}d_{l\varphi}(j)z_{\varphi k}=\sum_{j=1}^n
z_{sj}\delta_{lj}z_{jk}=z_{sl}z_{lk}=\sum_{j=1}^n
\omega_{lj}(s)z_{jk},
\daqm
and it implies that $\sum_{j=1}^n z_{sj}D(j)Z=\Omega(s)Z$. Symmetrically,
$\sum_{j=1}^n z_{sj}ZD(j)=Z^T\Omega(s)$. Equality $\sum_{j=1}^n
z_{sj}D(j)=\Omega(s)$ can be easily established. This completes the proof.

\end{proof}

We would like to note that in both cases we obtain the convergence to
rank deficient (singular) multivariate normal distributions.

Of course, one can use the joint confidence intervals for the CLT approximations
to estimate the quality of top-$k$ list or basket. However, it appears that we can
propose more efficient methods. Let us consider as an example mutual ranking of
two elements $k$ and $l$ from a list. For illustration purpose, assume that the
elements are independent and have the same variance. Suppose that we apply some
version of CLT approximation. Then, we need to compare two normal random variables
$Y_k$ and $Y_l$ with means $\pi_k$ and $\pi_l$, and with the same variance $\sigma^2$.
Without loss of generality we assume that $\pi_k > \pi_l$. Then, it can be shown
that one needs twice as more experiments to guarantee that the random variable
$Y_k$ and $Y_l$ inside their confidence intervals with the confidence level
$\alpha$ than to guarantee that $P\{Y_k \ge Y_l\}=\alpha$. Thus, it is more
beneficial to look at the order of elements rather than their absolute values.
We shall pursue this idea in more detail in the ensuing sections.


\section{Convergence based on order}
\label{sec:order}

For the two introduced Monte Carlo methods we would like to calculate or to
estimate a probability that after a given number of steps we correctly obtain
top-$k$ list or top-$k$ basket. Namely, we need to calculate the probabilities
$P\{L_1 > \dots > L_k > L_j, \forall j > k\}$ and $P\{L_i > L_j, \forall i,j : i
\le k < j \}$ respectively, where $L_k$, $k\range{1}{n}$, can be either the
Monte Carlo estimates of the ranked elements or their CLT approximations. We
refer to these probabilities as the ranking probabilities and we refer to
complementary probabilities as misranking probabilities \cite{BID05}.
Even though, these probabilities are easy to define, it turns out that
because of combinatorial explosion their exact calculation is infeasible
for non-trivial cases.

We first propose to estimate the ranking probabilities of top-$k$ list and top-$k$ basket
with the help of Bonferroni inequality \cite{GS96}. This approach works for reasonably large
values of $m$.

\subsection{Estimation by Bonferroni inequality}
Drawing correctly the top-$k$ basket is defined by the event
$$
\bigcap_{i \le k < j} \{ L_i > L_j \}.
$$
Let us apply to this event the Bonferroni inequality
$$
P\left\{ \bigcap_s A_s \right\} \ge 1 - \sum_s P\left\{ \bar{A}_s \right\}.
$$
We obtain
$$
P\left\{ \bigcap_{i \le k < j} \{ L_i > L_j \} \right\}
\ge 1 - \sum_{i \le k < j} P\left\{ \overline{\{ L_i > L_j \}} \right\}.
$$
Equivalently, we can write the following upper bound for the misranking probability
\begin{equation}\label{eq:ubmisrank}
1-P\left\{ \bigcap_{i \le k < j} \{ L_i > L_j \} \right\}
\le \sum_{i \le k < j} P\left\{ L_i \le L_j \right\}.
\end{equation}
We note that it is very good that we obtain an upper bound in the above
expression for the misranking probability, since the upper bound will provide a
guarantee on the performance of our algorithms. Since in the MC End Point method
the distribution of end points is multinomial (see (\ref{chap:top-k:multinom})),
the probability $P\left\{ L_i \le L_j \right\}$ is given by
\begin{equation}\label{eq:combmisrank}
P\{L_i \le L_j\}=
\end{equation}
$$
\sum_{l_i+l_j \le m, \ l_i \le l_j} \frac{m!}{l_i!l_j!(m-l_i-l_j)!} \pi_i^{l_i} \pi_j^{l_j}
(1-\pi_i-\pi_j)^{m-l_i-l_j}.
$$
The above formula can only be used for small values of $m$. For large values of $m$,
we can use the CLT approximation for the both MC methods. To distinguish between the original
number of hits and its CLT approximation, we use $L_j$ for the original number of hits at node $j$
and $Y_j$ for its CLT approximation. First, we obtain a CLT based
expression for the misranking probability for two nodes $P\left\{ Y_i \le Y_j \right\}$.
Since the event $\left\{ Y_i \le Y_j \right\}$ coincides with the event
$\left\{ Y_i-Y_j \le 0 \right\}$ and a difference of two normal random variables is
again a normal random variable, we obtain
$$
P\left\{ Y_i \le Y_j \right\} = P\left\{ Y_i-Y_j \le 0 \right\} = 1 - \Phi (\sqrt{m} \rho_{ij}),
$$
where $\Phi(\cdot)$ is the cumulative distribution function for the standard
normal random variable and
$$
\rho_{ij} = \frac{E[Y_i]-E[Y_j]}{\sqrt{\sigma^2(Y_i)-2\mbox{cov}(Y_i,Y_j)+\sigma^2(Y_j)}}.
$$
For large $m$, the above expression can be bounded by
$$
P\left\{ Y_i \le Y_j \right\} \le
\frac{1}{\sqrt{2\pi}}e^{-\frac{\rho_{ij}^2}{2}m} $$
Since the misranking probability for two nodes $P\left\{ Y_i \le Y_j
\right\}$ decreases when $j$ increases, we can write
$$
1-P\left\{ \bigcap_{i \le
 k < j} \{ Y_i > Y_j \} \right\} \le
$$
$$
\sum_{i=1}^k
\left( \sum_{j=k+1}^{j^*} P\left\{ Y_i \le Y_j \right\} + \sum_{j=j^*+1}^n P\left\{ Y_i \le Y_{j^*} \right\} \right),
$$
for some $j^*$. This gives the following upper bound
$$
1-P\left\{ \bigcap_{i \le k < j} \{ Y_i > Y_j \} \right\} \le
$$
\begin{equation}\label{chap:top-k:eq:upperboundexp}
\sum_{i=1}^k \sum_{j=k+1}^{j^*} (1 - \Phi (\sqrt{m} \rho_{ij}))
+ \frac{n-j^*}{\sqrt{2\pi}} \sum_{i=1}^k e^{-\frac{\rho_{ij^*}^2}{2}m}.
\end{equation}

Since we have a finite number of terms in the right hand side of expression
\eqref{chap:top-k:eq:upperboundexp}, we conclude that

\begin{theorem}
The misranking probability of the top-$k$ basket tends to zero with geometric rate, that is,
$$
1-P\left\{ \bigcap_{i \le k < j} \{ Y_i > Y_j \} \right\} \le C a^m,
$$
for some $C>0$ and $a \in (0,1)$.
\end{theorem}

We note that $\rho_{ij}$ has a simple expression in the case of the multinomial distribution
$$
\rho_{ij} = \frac{\pi_i-\pi_j}{\sqrt{\pi_i(1-\pi_i)+2\pi_i\pi_j+\pi_j(1-\pi_j)}}.
$$
For MC Complete Path $\sigma^2(Y_i)=\Sigma_{ii}(s)$ and $\mbox{cov}(Y_i,Y_j)=\Sigma_{ij}(s)$
where $\Sigma_{ii}(s)$ and $\Sigma_{ij}(s)$ can be calculated by (\ref{chap:top-k:eq:covar_matrix_expr}).

The Bonferroni inequality for the top-$k$ list gives
$$
P\{Y_1 > \dots > Y_k > Y_j, \forall j > k\} \ge
$$
$$
1 - \sum_{1 \le i \le k-1}
P\{Y_i \le Y_{i+1}\} - \sum_{k+1 \le j \le n} P\{Y_k \le
Y_{j}\}.
$$
Using misranking probability for two elements, one can obtain more informative bounds
for the top-$k$ list as was done above for the case of top-$k$ basket. For the misranking
probability of the top-$k$ list we also have a geometric rate of convergence.

\bigskip

\noindent {\bf Illustrating example (cont.):}
In Figure~\ref{fig:jjplayerBonfEP} we plot Bonferroni bound for the misranking probability
given by (\ref{eq:ubmisrank}) with the CLT approximation for the pairwise misranking probability.
We note that the Bonferroni inequality provides quite conservative estimation for the necessary
number of MC runs. Below we shall try to obtain a better estimation.

\begin{figure}
    \centering
    \includegraphics[height=2in,width=3in]{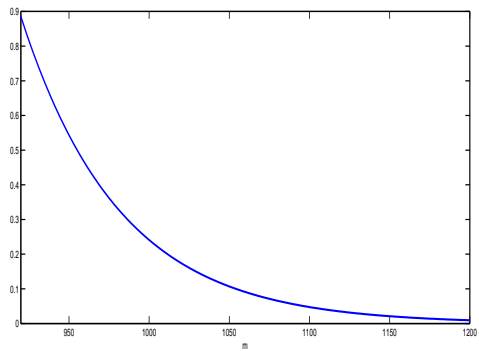}
    \caption{Bonferroni bound for MC End Point for the seed node
    {\tt Jim Jackson (ice hockey)} and top-9 basket.}
    \label{fig:jjplayerBonfEP}
\end{figure}

\subsection{Approximation based on order statistics}

We can obtain more insight on the convergence based on order with the help of order statistics \cite{DN03}.

Let us denote by $X_t \in \{1,...,n\}$ the node hit by the random walk $t$,
and let $X$ have the same distribution as $X_t$.
Now let us consider the $s$-th order statistic of the random
variables $X_t$, $t=1,...,m$. We can calculate
its cumulative distribution function as follows:
$$
P\{X_{(s)} \le k\} =
\sum_{j=0}^{m-s} {m\choose j}  P\{X>k\}^j P\{X \le k\}^{m-j}
$$
\begin{equation}
=1- \sum_{i=0}^{s-1} {m\choose i} P\{X \le k\}^{i} (1-P\{X \le k\})^{m-i}
\label{eq:sorder}
\end{equation}


It is interesting to observe that $P\{X_{(s)} \le k\}$ depends on the Personalized
PageRank distribution only via $P\{X \le k\} = \pi_1+...+\pi_k$ showing insensitivity
property with respect to the distribution's tail.

Next we notice that a reasonable minimal value of $m$ corresponds to the case
when the elements of the top-$k$ basket obtain $r$ or more hits with a high probability
and the other elements outside the top-$k$ basket will have very small probability
of $r$-times hit. Thus, the probability $P\{X_{(rk)} \le k\}$ should be reasonably high
and the probability of hitting $r$ times  the elements outside the top-$k$ basket should be
small. The probability to hit the element $j$ at least $r$ times is given by
\begin{equation}\label{eq:telement}
P\{Y_j \ge r\} = 1 - \sum_{\nu=0}^{r-1}{{m}\choose{\nu}}\pi_j^\nu(1-\pi_j)^{m-\nu}.
\end{equation}
Hence, choosing $m$ for the fast detection of the top-$k$ basket we need to
satisfy two criteria: (i) $P\{X_{(rk)} \le k\} \ge 1-\varepsilon_1$, and
(ii) $P(Y_j \ge r) \le \varepsilon_2$ for $j>k$. The probability in (i)
and $P(Y_j \ge r)$ in (ii) are both increasing with $m$. However, we have observed
(see the illustrating example next) that for a given $m$ the probabilities given
in (\ref{eq:sorder}) and (\ref{eq:telement}) drop drastically with $r$.  Thus,
we hope to be able to find a proper balance between $m$ and $r$ for a reasonably
small value of $r$.


We can further improve the computational efficiency for order statistics
distribution (\ref{eq:sorder}) with the help of incomplete Beta function
as suggested in \cite{AN05}. Namely, in our case we have
\begin{equation}
\label{eq:sorderbeta}
P\{X_{(s)} \le k\} = I_{P\{X \le k\}}(s,m-s+1),
\end{equation}
where
$$
I_x(a,b) = \frac{1}{B(a,b)} \int_0^x y^{a-1} (1-y)^{b-1} dy
$$
is the incomplete Beta function.


\bigskip

\noindent {\bf Illustrating example (cont.):}
We first consider the seed node {\tt Jim Jackson (ice hochey)}.
In Figure~\ref{fig:osk9t10jjplayer} we plot the probabilities given by (\ref{eq:sorder}) and
(\ref{eq:telement}) for $m \le 2000$, $r=5$, and $k=9$. For instance,
if we take $m=250$, $P\{X_{(45)} \le 9\}=0.9999$ and $P\{Y_{10} \ge 5\}=4.29\times10^{-5}$.
Thus, with very high probability we collect 45 hits inside the top-9 basket and
the probability for the 10-th element to receive more than or equal to 5 hits is very
small. Figure~\ref{fig:jjplayermcep} confirms that taking $m=250$ is largely enough
to detect the top-9 basket. Suppose now that we want to detect the top-10 basket.
Then, Figure~\ref{fig:osk10t11jjplayer} corresponding to $m \le 10000$ , $r=18$ and $k=10$ suggests that to obtain
correctly top-10 basket with high probability we need to spend about four times more
operations than for the case of the top-9 basket. Here we already see an illustration
to the ``80/20 rule'' which we discuss more in the next section.
Now let us consider the seed node {\tt Michael Jackson}. In Figure~\ref{fig:osk10t11mjsinger} we plot
the probabilities given by (\ref{eq:sorder}) and (\ref{eq:telement})
for $m \le 100000$, $r=57$, $k=10$ and $j=11$. We have $P\{X_{(570)} \le 10\}=0.9717$ and
$P\{Y_{11} \ge 57\}=0.2338$. Even though there is a significant chance to get some erroneous
elements in the top-10 list, as Figure~\ref{fig:osk10t100mjsinger} suggests we get ``high quality''
erroneous elements. Specifically, we have $P\{Y_{100} \ge 57\}=0.0077$ and
$P\{Y_{500} \ge 57\}=8.8\times10^{-47}$.

\begin{figure}
    \centering
    \includegraphics[scale=0.5]{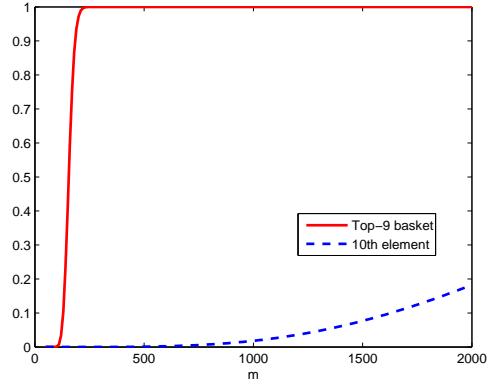}
    \caption{Evaluations based on order statistics for the seed node {\tt Jim Jackson (ice hockey)}:
    $P\{X_{(rk)} \le k\}$ (solid line) and $P\{L_j \ge r\}$ (dash line), $k=9$, $j=10$, $r=5$.}
    \label{fig:osk9t10jjplayer}
\end{figure}

\begin{figure}
    \centering
    \includegraphics[scale=0.5]{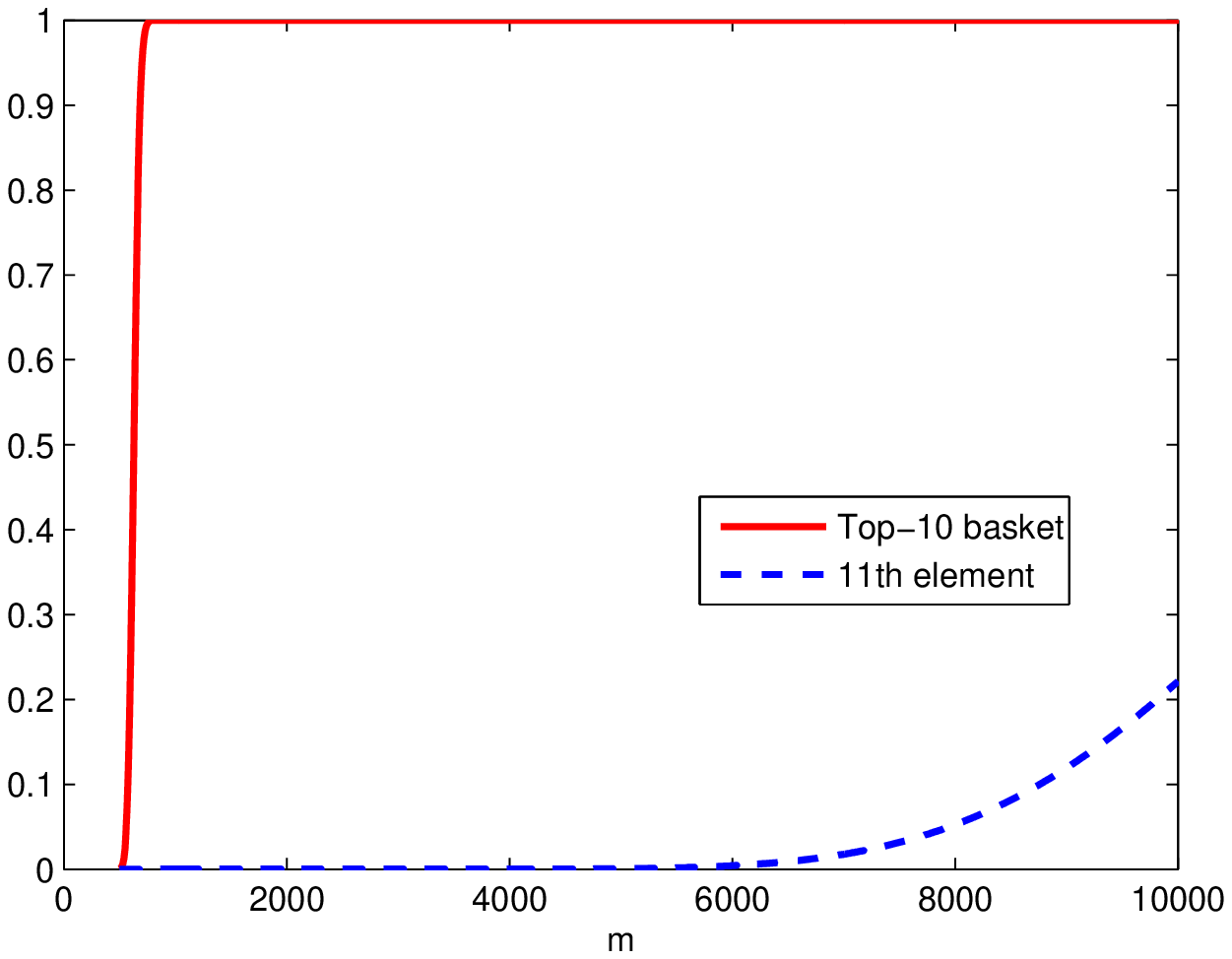}
    \caption{Evaluations based on order statistics for the seed node {\tt Jim Jackson (ice hockey)}:
    $P\{X_{(rk)} \le k\}$ (solid line) and $P\{L_j \ge r\}$ (dash line), $k=10$, $j=11$, $r=18$.}
    \label{fig:osk10t11jjplayer}
\end{figure}

\begin{figure}
    \centering
    \includegraphics[scale=0.5]{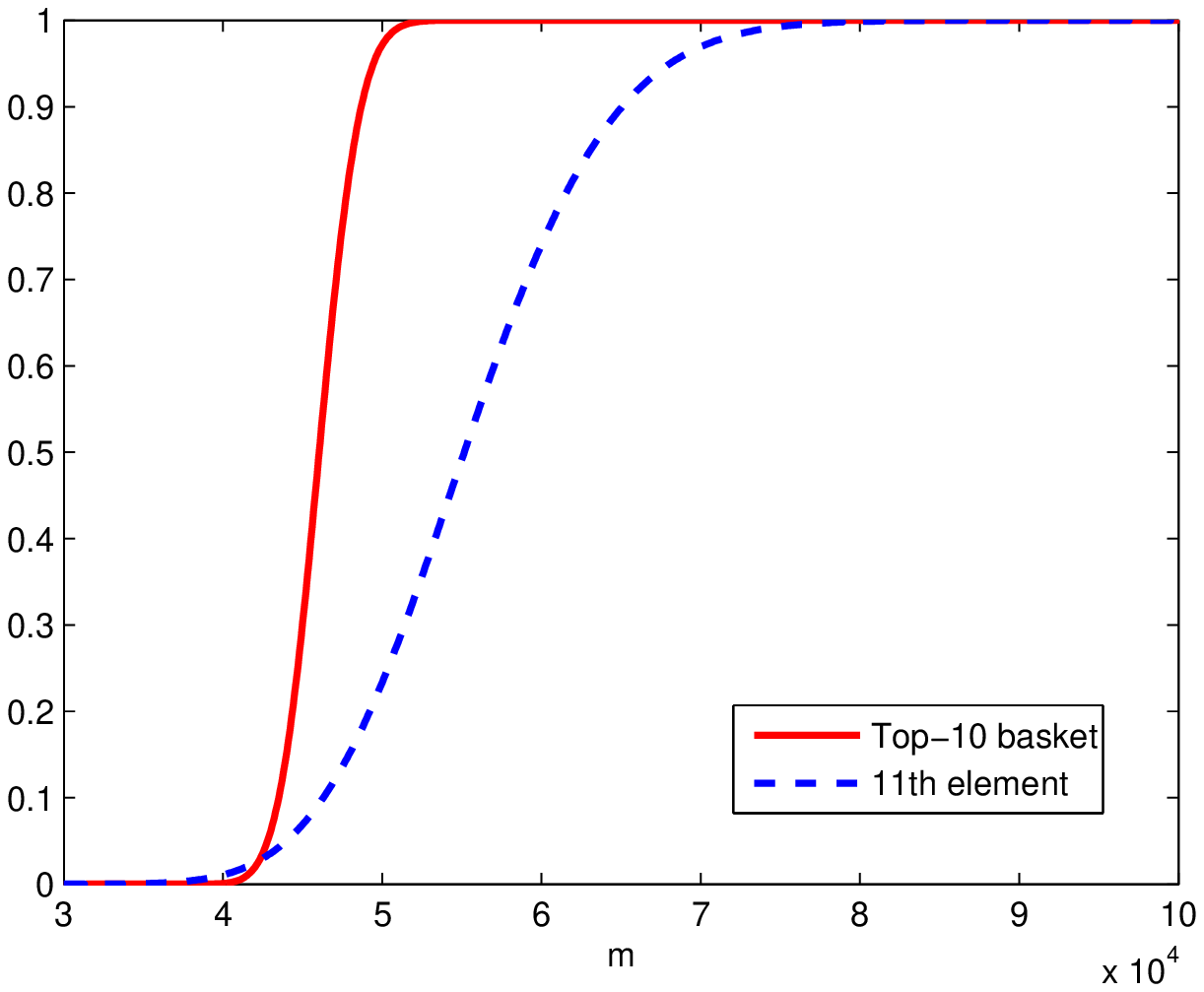}
    \caption{Evaluations based on order statistics for the seed node {\tt Michael Jackson}:
    $P\{X_{(rk)} \le k\}$ (solid line) and $P\{L_j \ge r\}$ (dash line), $k=10$, $j=11$, $r=57$.}
    \label{fig:osk10t11mjsinger}
\end{figure}

\begin{figure}
    \centering
    \includegraphics[scale=0.5]{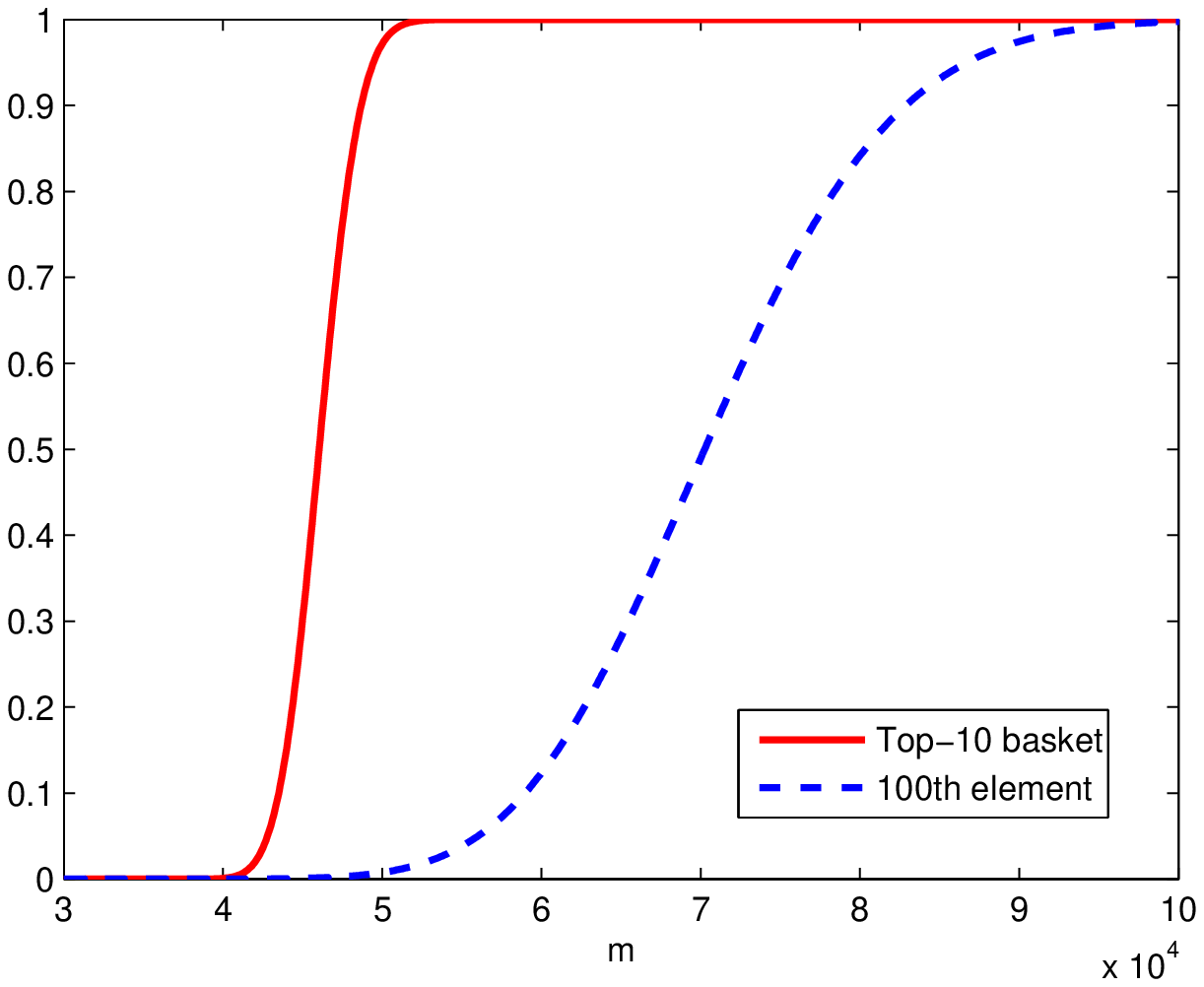}
    \caption{Evaluations based on order statistics for the seed node {\tt Michael Jackson}:
    $P\{X_{(rk)} \le k\}$ (solid line) and $P\{L_j \ge r\}$ (dash line), $k=10$, $j=100$, $r=57$.}
    \label{fig:osk10t100mjsinger}
\end{figure}

\section{Solution relaxation}
\label{sec:relax}

In this section we analytically evaluate the average number of correctly
identified top-$k$ nodes. We use the relaxation by allowing this number to be smaller than $k$. Our goal is to provide a
mathematical evidence for the observed ``80/20 behavior'' of the algorithm: 80 percent of the top-$k$ nodes are identified correctly in a very short time. Accordingly, we evaluate the number of experiments $m$ for obtaining high quality top-$k$ lists.

Let $M_0$ be a number of correctly identified elements in the top-$k$ basket.
In addition, denote by $K_i$ the number of nodes ranked not lower than $i$. Formally,
\[K_i=\sum_{j\ne i}1\{L_j\ge L_i\},\quad i=1,\ldots,k.\]
Clearly, placing node $i$ in the top-$k$ basket is equivalent to the event $[K_i<k]$, and thus we obtain
\begin{align}
\label{eq:em0}
E(M_0)&=E\left(\sum_{i=1}^k 1\{\mbox{$K_i<k$}\}\right)=\sum_{i=1}^kP(\mbox{$K_i<k$}).
\end{align}

To evaluate $E(M_0)$ by (\ref{eq:em0}) we need to compute the probabilities $P(K_i<k)$ for
$i=1,\ldots, k$. Direct evaluation of these probabilities is computationally intractable. A Markov chain approach based on the representations from \cite{journal/SSRNeLibrary/Corrado2007} is more efficient, but this method, too, resulted in extremely demanding numerical schemes in realistic scenarios. Thus, to characterize the algorithm performance, we suggest to use two simplification steps: {\it approximation} and {\it Poissonisation}.

{\it Poissonisation} is a common technique for analyzing
occupancy measures~\cite{Gnedin2007}.
Clearly, the End Point algorithm is nothing else but an occupancy scheme where each independent experiment
(random walk) results in placing one ball (visit) to an urn (node
of the graph). Under Poissonisation, we assume that the number of random walks is not a fixed
value $m$ but a Poisson random variable $M$ with mean $m$. In this scenario, the number $Y_j$ of visits to page $j$ has a Poisson distribution with parameter $m\pi_j$ and is independent of $Y_i$ for $i\ne j$. Because the number of hits in the Poissonised
model is different from the number of original hits, we use the notation $Y_i$ instead of $L_j$. Poissonisation simplifies
the analysis considerably due to the imposed independence of the $Y_j$'s.

Next to Poissonisation, we also apply {\it approximation} of $M_0$ by a closely related measure $M_1$:
\[M_1:=k-\sum_{i=1}^{k}(K'_i/{k}),\]
where $K'_i$ denotes the number of pages outside the top-$k$ that are ranked
higher than node $i=1,\ldots,k$. The idea behind $M_1$ is as follows: $K_i'$ is the number of mistakes with respect to node $i$ that lead to errors in the identified top-$k$ list. The sum in the definition of $M_1$ is the average number of such mistakes with respect to each of the top-$k$ nodes.

Two properties of $M_1$ make it more tractable than $M_0$. First, the average value of $M_1$ is defined as
\[E(M_1)=k-\frac{1}{k}\sum_{i=1}^{k}E(K'_i),\]
which involves only the average values of $K_i'$ and not their distributions. Second, $K_i'$ involves only the nodes outside of the top-$k$ for each $i=1,\ldots,k$, and thus we can make use of the following convenient measure $\mu(y)$:
\[\mu(y):=E(K'_i|Y_i=y)=\sum_{j=k+1}^nP(Y_j\ge y),\;i=1,\ldots,k,\]
which implies
\[E(K_i')=\sum_{y=0}^\infty P(Y_i=y)\mu(y),\;i=1,\ldots,k.\]
Therefore, we obtain the following expression for $E(M_1)$:
\begin{equation}\label{eq:em4}
E(M_1)=k-\frac{1}{k}\sum_{y=0}^{\infty}\mu(y)\sum_{i=1}^kP(Y_i=y).\end{equation}

\bigskip

\noindent {\bf Illustrating example (cont.):}
Let us calculate $E(M_1)$ for the top-10 basket corresponding to the seed
node {\tt Jim Jackson (ice hockey)}. Using formula (\ref{eq:em4}), for
$m=8\times10^3;10\times10^3;15\times10^3$
we obtain $E(M_1)=7.75;9.36;9.53$. It took $2000$ runs to move from
$E(M_1)=7.75$ to $E(M_1)=9.36$, but then it needed $5000$ runs to advance from
$E(M_1)=9.36$ to $E(M_1)=9.53$. We see that we obtain quickly 2-relaxation
or 1-relaxation of the top-10 basket but then we need to spend a significant
amount of effort to get the complete basket. This is indeed in agreement with
the Monte Carlo runs (see e.g., Figure~\ref{fig:jjplayermcep}).
In the next theorem we explain this ``80/20 behavior'' and provide indication for
the choice of $m$.


\begin{theorem}
\label{th:m_m4}

In the Poisonized End Point Monte Carlo algorithm, if all top-$k$ nodes receive at least $y=ma>1$ visits and $\pi_{k+1}=(1-\varepsilon)a$, where $\varepsilon>1/y$ then

(i) to satisfy $E(M_1)>(1-\alpha) k$ it is sufficient to have
\[\sum_{j=k+1}^n\frac{(m\pi_j)^y}{y!}\,e^{-m\pi_j}\left[1+\sum_{l=1}^{\infty}\frac{(m\pi_{j})^l}{(y+1)\cdots (y+l)}\right]<\alpha k,\]
and

(ii) statement (i) is always satisfied if
\begin{equation}
\label{eq:upper_bound_m}
m> 2a^{-1}\varepsilon^{-2}[-\log(\varepsilon \pi_{k+1} \alpha  k)].\end{equation}
\end{theorem}

{\bf Proof.}  (i) By definition of $M_1$, to ensure that $E(M_1)\le (1-\alpha) k$ it is sufficient that $E(K'_i| Y_i)\le \alpha k$ for each $Y_i\ge y$ and each $i=1,\ldots,k$. Now, (i) follows directly since for each $Y_i\ge y$ we have
$E(K'_i| Y_i)\le \mu(y)$ and by definition of $\mu(y)$ under Poissonisation we have
\begin{align}
\label{eq:mu(y)}
 \mu(y)&=\sum_{j=k+1}^n\frac{(m\pi_j)^y}{y!}\,e^{-m\pi_j}\left[1+\sum_{l=1}^{\infty}\frac{(m\pi_{j})^l}{(y+1)\cdots (y+l)}\right].
\end{align}

To prove (ii), using (\ref{eq:mu(y)}) and the conditions of the theorem, we obtain:
\begin{align}
\nonumber
\mu(y)
&\le \sum_{j=k+1}^n\frac{(m\pi_j)^y}{y!}\,e^{-m\pi_j}\left[1+(1-\varepsilon)+(1-\varepsilon)^2+\cdots\right]\\
\nonumber
&= \frac{1}{\varepsilon}\sum_{j=k+1}^n\frac{(m\pi_j)^y}{y!}\,e^{-m\pi_j}
=\frac{1}{\varepsilon}\sum_{j=k+1}^n\pi_j\frac{m^y\pi_j^{y-1}}{y!}\,e^{-m\pi_j}\\
\nonumber
&\stackrel{\{1\}}{\le} \frac{1}{\varepsilon}\, \frac{m^y\pi_{k+1}^{y-1}}{y!}\,e^{-m\pi_{k+1}}\stackrel{\{2\}}{\le}\frac{1}{\varepsilon}\, \frac{1}{\pi_{k+1}}\left(\frac{m\pi_{k+1}}{y}\right)^ye^{y-m\pi_{k+1}}\\
\label{eq:mu_upper_bound}
&
=\frac{1}{\varepsilon \pi_{k+1}}\,\left[(1-\varepsilon)e^\varepsilon\right]^{ma}.
\end{align}
Here $\{1\}$ holds because $\sum_{j\ge k+1}\pi_j\le 1$ and $(m\pi_{j})^{y-1}/(y-1)!\exp\{-m\pi_{j}\}$ is maximal at $j=k+1$. The latter follows from the conditions of the theorem: $m\pi_{k+1}=(1-\varepsilon)y\le y-1$ when $\varepsilon>1/y$. In
\{2\} we use that $y!\ge y^y/e^y$.

Now, we want the last expression in (\ref{eq:mu_upper_bound}) to be smaller than $\alpha\,k$. Solving for $m$, we get:
\[ma(\log(1-\varepsilon)+\varepsilon)<\log(\varepsilon \pi_{k+1}\alpha\,k).\]
Note that the expression under the logarithm on the right-hand side is always smaller than 1 since $\alpha<1$, $\varepsilon<1$ and $k\pi_{k+1}<1$. Using  $(\log(1-\varepsilon)+\varepsilon)=-\sum_{k=2}^\infty\varepsilon^k/k\ge -\varepsilon^2/2$, we obtain (ii). \qed

From (i) we can already see that the 80/20 behavior of $E(M_1)$ (and, respectively, $E(M_0)$) can be explained mainly by the fact that $\mu(y)$ drops drastically with $y$ because the Poisson probabilities decrease faster than exponentially.

The bound in (ii) shows that $m$ should be rougthly of the order $1/\pi_{k}$. The term $\varepsilon^{-2}$ is not defining since $\varepsilon$ does not need to be small. For instance, by choosing $\varepsilon=1/2$ we can filter out the nodes with Personalized PageRank not higher than $\pi_k/2$. This often may be sufficient in applications. Obviously, the logarithmic term is of a smaller order of magnitude.

We note that the bound in (ii) is very rough because in its derivation we replaced $\pi_j$, $j>k$, by their maximum value $\pi_{k+1}$. In realistic examples, $\mu(y)$ will be much smaller than the last expression in (\ref{eq:mu_upper_bound}), which allows for $m$ much smaller than in (\ref{eq:upper_bound_m}). In fact, in our examples good top-$k$ lists are obtained if the algorithm is terminated at the point when for some $y$, each node in the current top-$k$ list has received at least $y$ visits while the rest of the nodes have received at most $y-d$ visits, where $d$ is a small number, say $d=2$. Such choice of $m$ satisfies (i) with reasonably small $\alpha$. Without a formal justification, this stopping rule can be understood since, roughtly, we have $m\pi_{k+1}=ma(1-\varepsilon)\approx ma-d$, which results in a small value of $\mu(y)$.


\section{Application to Name Disambiguation}
\label{sec:disamb}

In this section we apply Personalized PageRank computed using Monte-Carlo method to Person Name Disambiguation problem.
In the context of Web search when a user wants to retrieve information about a person by his/her name, search engines
typically return Web pages which contain the name but can refer to different persons. Indeed, person names are highly
ambiguous, according to US Census Bureau approximately 90,000 names are shared by 100 million people. Approximately
$5-10\%$ of search queries contain person name \cite{weps2-testbed}. To assist a user in finding the target person
many studies have been done, in particular, within WePS initiative (http://nlp.uned.es/weps/weps-3/).

In our approach we disambiguate the referents with the help of the Web graph structure. We define a \textit{related page}
as the one that addresses the same topic or one of the topics mentioned in a \textit{person page} - a page that contains person name. Kleinberg in \cite{hits-art} has given an illuminating example that ambiguous senses of the query can be separated on a query-focused subgraph. In our context, focused subgraph is analogous to a graph of Web search result pages with their graph-based neighbourhood represented by forward and backward links. Therefore, we can expect that for ambiguous person name query densely linked parts
of the subgraph form clusters corresponding to different individuals.

The major problem of applying HITS algorithm and other community discovery methods to WePS dataset consists in the lack of information about the Web graph structure. Personalized PageRank can be used to detect related pages of the target page. Our theoretical and experimental results show that, quite opportunely, Monte-Carlo method is a fast way to approximate Personalized PageRank in a local manner, i.e., using only page forward-links. In such a case global backward-link neighbours are usually missing. Therefore, generally we cannot expect neighbourhoods of two pages referring to one person to be interconnected. Nevertheless, we found useful to examine content of related pages.

In the following we will briefly describe our approach, further details will be published soon. With this approach we participated in WePS-3 evaluation campaign.

\subsection{System Description}
\subsubsection*{Web Structure based Clustering.}
It the first stage, we cluster person pages appeared in search results based on the Web structure. Thereto, we determine related pages
of each person page using Personalized PageRank. To avoid negative effect of purely navigational links we perform random walk of the Monte-Carlo computation on links to pages with different host name than the current host. We estimate the top-k list of related pages for each target page. In experiments we have used two values of $k = \left\{8,16\right\}$ and also two settings of Personalized PageRank computation: the damping factor $c$ equal to $\left\{0.2, 0.3\right\}$ respectively.

In the following step two Web pages that contain the name are merged in one cluster if they share some related pages. Since the whole link structure of the Web was unknown to us, the resulted Web structure clustering is limited to local forward-link neighbourhood of pages. We therefore appeal to the content of the pages in the next stage.

\subsubsection*{Content based Clustering.}
In the second stage, the rest of the pages that did not show any link preference are clustered based on the content. With this goal in mind,
we apply a preprocessing step to all pages including person pages and pages related to them. Next, for each of these page we build a vector of words with corresponding frequency score ($tf$) in the page. After that, we use a re-weighting scheme as follows. The word $w$ score at person page $t$, $tf(t,w)$, is updated at each related page $r$ in the following way: $tf'(t,w) = tf(t,w) + tf(t,w)*tf(r,w)$, where $r$
is a page in related pages set of person page $t$ and person page $t$ is a page obtained from search results. This step resembles voting process. Words that appear in related pages get promoted and thus, random word scores found in the person page are lowered. At the end, vector is normalized and top 30 frequent terms are taken as a person page profile.

Finally, we apply HAC algorithm on the basis of Web structure clustering to the rest of the pages. Specifically, we use average-linkage HAC with the cosine measure of similarity. The clustering threshold for HAC algorithm was determined manually.

\subsection{Results}
During the evaluation period of WePS-3 campaign we have experimented with the number of related pages and the type of content extracted from pages. We have chosen to combine a small number of related pages with the full content of the page and, oppositely, a large number of related pages with small extracted content. We carried out the following runs.

\textbf{PPR8 (PPR16)}: top 8 (16) related pages were computed using Personalized PageRank, the full (META tag) content of the Web page was used in the HAC step.

Our methods achieved second place performance at WePS-3 campaign. We received values of metrics for PPR8 (PPR16) runs as follows: \textbf{0.7(0.71); 0.45(0.43); 0.5(0.49)} for BCubed precision (BP), recall (BR) and harmonic mean of BP and BR (F-0.5) respectively. The run PPR16 has shown slightly worse performance compared to the PPR8 run. The results have demonstrated that our methods are promising and
suggest future research on using Web structure for the name disambiguation problem.


\section*{Acknowledgments}
We would like to thank Brigitte Trousse for her very helpful remarks and suggestions.

\end{document}